\documentclass[conference]{IEEEtran}
\IEEEoverridecommandlockouts
\usepackage{cite}
\usepackage{amsmath,amssymb,amsfonts}

\usepackage{algorithm}
\usepackage{algorithmic}

\usepackage{graphicx}
\usepackage{textcomp}
\usepackage{xcolor}
\usepackage{amsthm}
\usepackage{comment}
\usepackage{multirow}
\usepackage{subfigure}
\usepackage{caption}
\usepackage{tabularx}
\newcolumntype{L}[1]{>{\raggedright\arraybackslash}p{#1}}
\newcolumntype{C}[1]{>{\centering\arraybackslash}p{#1}}
\newcolumntype{R}[1]{>{\raggedleft\arraybackslash}p{#1}}

\def\BibTeX{{\rm B\kern-.05em{\sc i\kern-.025em b}\kern-.08em
    T\kern-.1667em\lower.7ex\hbox{E}\kern-.125emX}}
\begin{document}

\title{FI-GRL: Fast Inductive Graph Representation Learning via Projection-Cost Preservation
}

\author{
\IEEEauthorblockN{Fei Jiang\IEEEauthorrefmark{1}, Lei Zheng\IEEEauthorrefmark{2}, Jin Xu\IEEEauthorrefmark{1}, Philip S. Yu\IEEEauthorrefmark{2}\IEEEauthorrefmark{3}}
\IEEEauthorblockA{\IEEEauthorrefmark{1}Department of Computer Science, Peking University, Beijing, China\\
\IEEEauthorrefmark{2}Department of Computer Science, University of Illinois at Chicago, IL, US \\
\IEEEauthorrefmark{3}Institute for Data Science, Tsinghua University, Beijing, China.
\\
allen.feijiang@gmail.com, \{lzheng21, psyu\}@uic.edu, jxu@pku.edu.cn}
}

\maketitle

\begin{abstract}
Graph representation learning aims at transforming graph data into meaningful low-dimensional vectors to facilitate the employment of machine learning and data mining algorithms designed for general data. Most current graph representation learning approaches are transductive, which means that they require all the nodes in the graph are known when learning graph representations and these approaches cannot naturally generalize to unseen nodes. In this paper, we present a Fast Inductive Graph Representation Learning framework (FI-GRL) to learn nodes' low-dimensional representations. Our approach can obtain accurate representations for seen nodes with provable theoretical guarantees and can easily generalize to unseen nodes. Specifically, in order to explicitly decouple nodes' relations expressed by the graph, we transform nodes into a randomized subspace spanned by a random projection matrix. This stage is guaranteed to preserve the projection-cost of the normalized random walk matrix which is highly related to the normalized cut of the graph. Then feature extraction is achieved by conducting singular value decomposition on the obtained matrix sketch. By leveraging the property of projection-cost preservation on the matrix sketch, the obtained representation result is nearly optimal. To deal with unseen nodes, we utilize folding-in technique to learn their meaningful representations. Empirically, when the amount of seen nodes are larger than that of unseen nodes, FI-GRL always achieves excellent results. Our algorithm is fast, simple to implement and theoretically guaranteed. Extensive experiments on real datasets demonstrate the superiority of our algorithm on both efficacy and efficiency over both macroscopic level (clustering) and microscopic level (structural hole detection) applications.
\end{abstract}

\begin{IEEEkeywords}
Graph Representation Learning, Inductive Learning, Graph Mining
\end{IEEEkeywords}

\newtheorem{definition}{\textbf{Definition}}
\newtheorem{theorem}{\textbf{Theorem}}
\newtheorem{lemma}{\textbf{Lemma}}

\section{Introduction}
Graphs with nodes representing entities and edges representing relationships between entities are ubiquitous in various research fields. However, since a graph is naturally expressed in a node-interrelated way, it is exhausted to directly design different complicated graph algorithms for various kinds of mining and analytic purposes on graph data. Graph representation learning (also known as, graph embedding or network embedding) aims at learning from graph data to obtain low-dimensional vectors to represent nodes without losing much information contained in the original graph. Afterwards, one can apply bunches of off-the-shelf machine learning and data mining algorithms designed for general data on various important applications (e.g., clustering \cite{qiu2007clustering}, structural hole detection \cite{he2016joint}, link prediction \cite{lin2015learning}, visualization \cite{maaten2008visualizing}, etc).


Due to its ability in facilitating graph analysis, graph representation learning has drawn researchers' attentions from machine learning and data mining fields \cite{dong2017metapath2vec,ma2017multi,tu2017cane}. Most of these works are focusing on static networks, such as explicitly preserving local or high-order proximity \cite{tang2015line,wang2016structural,cavallari2017learning}, learning representations using truncated walks \cite{perozzi2014deepwalk,tu2016max}; using matrix factorization technique to obtain latent vectors \cite{yang2017fast,yang2015network,ma2017multi}, incorporating heterogeneous information \cite{tu2017cane,dong2017metapath2vec}, etc. Most of these methods are lack of theoretical support to produce reliable representation results. These methods are inherently transductive, which means that they are acting as black boxes that only care about learning representations but do not have an internal mechanism to naturally generalize to unseen nodes. 



\begin{figure}[!tb]
\centering
\includegraphics[width=\linewidth]{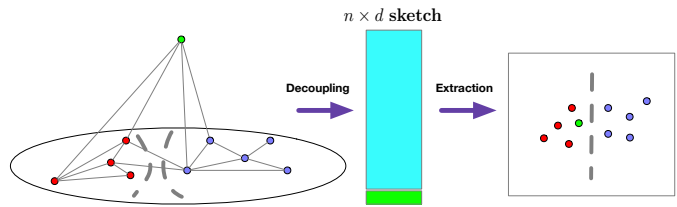}
\caption{A simple illustration of the intuition of our framework. Blue and red points can be assigned to two clusters respectively and the green node is an unseen node. 
Our approach assigns the unseen node to an appropriate location with preserving clustering membership and local proximity.
}
\label{fig:intuition}
\end{figure}

Our goal of graph representation learning is to design a fast and flexible framework that can preserve important graph topological information (e.g., clustering membership, node similarity, etc) with provable theoretical guarantees and can be naturally generalize to unseen nodes.
In this paper, in order to achieve this goal, the Fast Inductive Graph Representation Learning (FI-GRL) framework is proposed. FI-GRL consists of two stages: decoupling and feature extraction. The intuition of this architecture is illustrated in Figure \ref{fig:intuition}. The first stage is designed for decoupling nodes' relations by utilizing an oblivious algorithm, Johnson-Lindenstrauss random projection. 
For a graph $\mathcal{G}$, this stage generates a matrix sketch $\mathbf{M}\in\mathbb{R}^{n\times d}$, where $n$ is the number of nodes of graph $\mathcal{G}$ and $d$ is the sketch size, a parameter that can be automatically determined by our approach and is much smaller than $n$. The matrix sketch $\mathbf{M}$ approximates $\mathcal{G}$ well with theoretical guarantees. The second stage extracts meaningful feature contained in $\mathbf{M}$ by low rank approximation. Dimension reduction is achieved collaboratively in both of these two stages. The resulting representations by this framework are theoretically guaranteed to perform well on constrained low rank approximation tasks (e.g., $k$-means clustering). The main contributions of this paper are summarized as follows:
\begin{itemize}

\item \textbf{Architecture and Randomization}: The proposed framework FI-GRL is fast and flexible enough to handle large graphs. Since the decoupling stage adopts an oblivious randomized algorithm, nodes can be processed sequentially with a single pass and without storing the entire graph. The matrix sketch is much smaller that can be dealt with much faster by the feature extraction stage. Moreover, the first stage is projection-cost preserving, which makes sure that the resulting representations extracted by the second stage are optimal up to an approximation ratio $\epsilon$. As far as we know, this is the first time that randomized algorithms are introduced to deal with the graph representation learning problem. 

\item  \textbf{Theoretical Analysis}: We analyze our algorithm theoretically. We prove the optimality of our algorithm in terms of the absolute difference of projection-cost between the learned representations and the desired representations, and also in terms of the absolute distance difference between their corresponding $k$-mean centroids in Theorem \ref{theorem:difference}. For the choice of the parameter (the sketch size $d$), we give a theoretical guidance in Section \ref{sec:parameter_complex} and empirical analysis in Section \ref{sec:para_empirical_sketch}.

\item \textbf{Inductive Learning}: Our two-stage framework can naturally generalize to learn representations of unseen nodes. We adopt an incremental singular value decomposition with folding-in technique on the matrix sketch to learn representations of unseen nodes. The empirical results reporting in Section \ref{sec:perfor_streaming} demonstrate the effectiveness of our method on unseen nodes.

\item \textbf{Empirical Study}: FI-GRL can produce graph representations of different accuracy for different levels of applications. In Section \ref{sec:results}, extensive experiments conducted on both macroscopic level (clustering) and microscopic level (structural hole detection) applications show the superiority of our framework in efficacy and efficiency.
\end{itemize}

\section{Preliminary and Problem Formulation} 
\label{sec:prelim}
A graph $\mathcal{G}=(V,E,\mathbf{W})$ is a tuple with three elements, $V$ denoting node set, $E$ denoting edge set and $\mathbf{W}$ denoting weighted adjacency matrix. Without ambiguity, we use the term graph and network interchangeably. In this paper, to facilitate the distinction, we use lowercase letters (e.g., $\lambda$) to denote scalars, bold lowercase letters (e.g., $\mathbf{x}$) to denote vectors, bold uppercase letters (e.g., $\mathbf{W}$) to denote matrices and calligraphic letters $\mathcal{G}$ to denote graphs. The symbol table is shown in Table \ref{tab:Symbol}. 

\begin{table}[htbp]
\small
\caption{\bf List of basic symbols}
\centering 
\begin{tabular}{|c|l|}
\hline
Symbol&Definition\\
\hline
$\mathcal{G}$ & Graph \\

$V,E,n,m$& Node, edge set and its corresponding volume\\



$N(v)$ & Neighbor set of node $v$ \\

$\mathbf{W},\mathbf{D}$& Weighted adjacency, diagonal degree matrices\\

$ \boldsymbol{\mathcal{L}}$ & Normalized random walk matrix\\

$||\boldsymbol{\cdot}||_F,||\boldsymbol{\cdot}||_2$& Frobenius norm and spectral norm \\

\hline
\end{tabular}
\label{tab:Symbol}
\end{table}

We first define the problem of graph representation learning as follows:
\begin{definition} 
\textbf{(Graph Representation Learning)} Given a graph $\mathcal{G}=(V,E,\mathbf{W})$, for a fixed dimension number $k\ll |V|$, graph representation learning aims at learning a map $f(v|\mathbf{W})= \mathbf{y} \in \mathbb{R}^k$ for every $v\in V$. 
\end{definition}

Graph representation learning will generate a low-dimensional vector for every node in the graph. The obtained vectors should preserve important information (e.g., clustering membership, node similarity, etc) hidden in the graph. The scenario is even more challenging when considering unseen nodes. We generally define the graph representation learning task for unseen nodes.

\begin{definition} 
\label{def:grlun}
\textbf{(Graph Representation Learning for Unseen Nodes)} A graph $\mathcal{G}'=(V',E',\mathbf{W}')$ is extended from graph $\mathcal{G}=(V,E,\mathbf{W})$ by adding nodes $V'\backslash V$ and associated edges $E'\backslash E$ after graph representations of $\mathcal{G}$ have been obtained. For each node $v\in V'\backslash V$, we learn a map $f(v)= \mathbf{y} \in \mathbb{R}^k$ without recomputing representations obtained thus far. 
\end{definition}

\section{Methodology}
\label{sec:method}
In this section, we first give some observations to show 1) why we choose the normalized random walk matrix $\boldsymbol{\mathcal{L}}$ as the initial input matrix of the graph. 2) how important tasks (e.g., $k$-means clustering) can be reduced to constrained low rank approximation problem on which our framework is theoretically guaranteed to work well. Then, we present our two-stage framework in detail and generalize it to unseen nodes. Finally, we provide a guidance on choosing parameters and show the complexity analysis. Our FI-GRL framework is illustrated in Figure \ref{fig:framework}.
\begin{figure*}[!tb]
\centering
\includegraphics[width=0.8\linewidth]{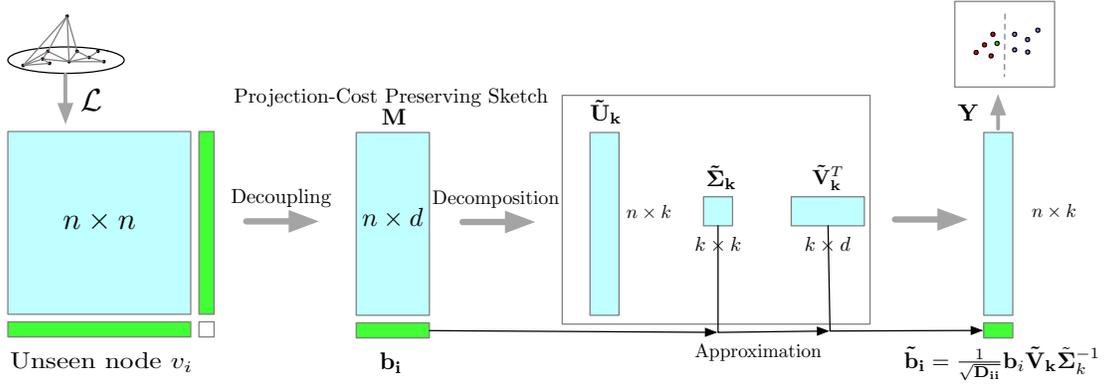}
\caption{An illustration of FI-GRL framework. This framework consists of two stages. The first stage is used for decoupling the nodes' relations by using a random projection algorithm. The second stage is used for feature extraction on matrix sketch $\mathbf{M}$. This two-stage framework provides a natural way for obtaining near-optimal representations for seen nodes and for generating approximate representations for unseen nodes.}
\label{fig:framework}
\end{figure*}
\subsection{Observations}
\subsubsection{Normalized Cut}
In literature \cite{shi2000normalized}, image segmentation is treated as a graph partition problem based on three metrics average association, normalized cut and average cut. It is demonstrated that normalized cut is a desired choice which seeks balance between finding clumps and finding split nodes. Normally, for graph $\mathcal{G}=(V,E,\mathbf{W})$ and two disjoint node sets $A,B\subset V$, $A\cap B=\varnothing$, the normalized cut is defined as:
\begin{equation}
Ncut=\frac{\Sigma_{u\in A,v\in B}\mathbf{W}(u,v)}{\Sigma_{u\in A,v\in V}\mathbf{W}(u,v)}+\frac{\Sigma_{u\in A,v\in B}\mathbf{W}(u,v)}{\Sigma_{u\in B,v\in V}\mathbf{W}(u,v)}.
\end{equation}
The problem of finding a graph partition with optimal normalized cut can be reduced to finding the generalized eigenvector problem of the following equation:
\begin{equation}
\label{eq:eigen_origin}
(\mathbf{D}-\mathbf{W})\mathbf{x}=\lambda \mathbf{Dx}.
\end{equation}
Then, applying k-means clustering on several smallest non-trivial eigenvectors $\mathbf{x}s$ will achieve the optimal graph partitions in terms of normalized cuts. Thus, the matrix related to Equation \ref{eq:eigen_origin} is a good choice for representing a graph.
\subsubsection{Constrained Low Rank Approximation}
\label{sec:clra}
To design a graph representation learning algorithm with preserving graph information such as clustering membership, we consider the constrained low rank approximation problem to which many important tasks, including $k$-means clustering, can be reduced. The constrained low rank approximation can be defined as
\begin{definition}
\label{de:constrained_LRA}
(\textbf{Constrained Low Rank Approximation}) For a matrix $\mathbf{A}\in \mathbb{R}^{n\times n}$ and any set $S$ of rank $k$ orthogonal projection matrices in $\mathbb{R}^{n\times n}$, constrained $k$ rank approximation tries to find
$$
\mathbf{P}^*=\underset{\mathbf{P}\in S}{\operatorname{argmin}} 
||\mathbf{A}-\mathbf{PA}||^2_F,
$$
where $||\mathbf{A}-\mathbf{PA}||^2_F$ is called as the projection-cost of $\mathbf{P}$.
\end{definition}

$k$-means clustering and approximate singular value decomposition (SVD) problems are constrained low rank approximation problems. More precisely, $k$-means clustering aims at dividing $n$ vectors of $\{\mathbf{a_1},\cdots,\mathbf{a_n}\}$, where $\mathbf{a_i}\in \mathbb{R}^b$, into $k$ clusters, ${C_1,\cdots,C_k}$. We denote the centroid of cluster $C_i$ as $\mathbf{u}_i$. The goal of $k$-means clustering is to minimize the following objective function:
$
\sum\limits_{j=1}^k \sum\limits_{\mathbf{a_i}\in C_j} ||\mathbf{a_i}-\mathbf{u_j}||^2_2.
$

We transform this objective function into the matrix form. We denote the cluster indicator matrix as matrix $\mathbf{X}\in \mathbb{R}^{n\times k}$, where $\mathbf{X}_{ij}=1/\sqrt{|C_j|}$ if $\mathbf{a}_i$ is divided into cluster $C_j$. So in the matrix form, $k$-means clustering is to minimize the following equation:
$
||\mathbf{A}-\mathbf{XX}^T \mathbf{A}||_F^2=\sum\limits_{j=1}^k \sum\limits_{\mathbf{a_i}\in C_j} ||\mathbf{a_i}-\mathbf{u_j}||^2_2.
$

Clearly, $\mathbf{XX}^T$ is a projection matrix that projects the points' vectors into their cluster centroids. Therefore, by definition \ref{de:constrained_LRA}, $k$-means clustering is a constrained low rank approximation problem with the set $S$ being the all possible projection matrix $\mathbf{XX}^T$, where $\mathbf{X}$ is the cluster indicator matrix. For SVD, it is also a constrained low rank approximation problem that tries to find the optimal rank $k$ approximation of $\mathbf{A}$ in the unconstrained set $S$ as all possible rank $k$ projection matrices. The optimal solution is its top $k$ left singular vectors $\mathbf{U}_k$. 

\subsection{Decoupling with Projection-Cost Preservation}
Suppose that if we can find a small matrix whose learned graph representations can achieve nearly the same results with that on the original $n\times n$ matrix, we will get lots of speed and space benefits. In this subsection, we show how we can obtain the small matrix (we call it matrix sketch) and also demonstrate its optimality. Formally, for a matrix $\mathbf{A}\in \mathbb{R}^{n\times n}$, we want to find a matrix sketch $\tilde{\mathbf{A}}\in \mathbb{R}^{n\times d}$ (where $d\ll n$) to approximate $\mathbf{A}$ well in a way that for constrained low rank approximation problem as in Definition \ref{de:constrained_LRA}, we can optimize projection matrix $\mathbf{P}\in \mathbb{R}^{n\times n}$ over the sketch $\tilde{\mathbf{A}}$ instead of optimizing $\mathbf{P}$ over $\mathbf{A}$.  
Firstly, we define the projection-cost preservation sketch as in \cite{cohen2015dimensionality}. 
\begin{definition}
\label{def:proj-cost preservation}
(\textbf{Rank $k$ Projection-Cost Preserving Sketch}) $\mathbf{\tilde{A}} \in \mathbb{R} ^{n \times d}$ is a rank $k$ projection-cost preserving sketch of $\mathbf{A}\in \mathbb{R}^{n\times n}$ with error $0<\epsilon <1 $ if, for all rank $k$ orthogonal projection matrices $\mathbf{P} \in \mathbb{R}^{n\times n}$, 
$$||\mathbf{A}-\mathbf{PA}||^2_F\leq ||\mathbf{\tilde{A}}-\mathbf{P\tilde{A}}||^2_F +c\leq (1+\epsilon)||\mathbf{A}-\mathbf{PA}||^2_F,$$
for some fixed non-negative constant $c$ that may depend on $\mathbf{A}$ and $\mathbf{\tilde{A}}$ but is independent of $\mathbf{P}$.
\end{definition}
This definition implies that the projection cost $||\mathbf{\tilde{A}}-\mathbf{P\tilde{A}}||^2_F$ of any projection $\mathbf{P}$ on $\mathbf{\tilde{A}}$ will be a good estimation of the projection cost $||\mathbf{A}-\mathbf{PA}||^2_F $ of the same projection over $\mathbf{A}$. The following lemma indicates that if $\mathbf{\tilde{A}}$ is a rank $k$ projection-cost preservation sketch of $\mathbf{A}$, one can optimize $\mathbf{\tilde{A}}$ to get the optimal projection matrix $\mathbf{P}$ to solve the constrained low rank approximation problem.

\begin{lemma}
\label{lemma:lrp_via_pcp}
Suppose $\mathbf{\tilde{A}}\in \mathbb{R}^{n\times d}$ is a projection-cost preservation sketch of $\mathbf{A}^{n\times n}$ with approximation ratio $\epsilon$ over the set $S$ of all rank $k$ projection matrices. Let $\mathbf{P}^*=\underset{\mathbf{P}\in S}{\operatorname{argmin}} 
||\mathbf{A}-\mathbf{PA}||^2_F$ and $\mathbf{\tilde{P}}^*=\underset{\mathbf{P}\in S}{\operatorname{argmin}} 
||\mathbf{\tilde{A}}-\mathbf{P\tilde{A}}||^2_F$. Then,
$
||\mathbf{A}-\mathbf{\tilde{P}^*A}||^2_F \leq (1+\epsilon)||\mathbf{A}-\mathbf{P^*A}||^2_F.
$
\end{lemma}
\begin{proof}
Since $\mathbf{\tilde{P}^*}$ is the optimal solution for $\mathbf{\tilde{A}}$, then 
\begin{equation}
\label{eq:lemma1}
||\mathbf{\tilde{A}}-\mathbf{\tilde{P}^*\tilde{A}}||^2_F\leq ||\mathbf{\tilde{A}}-\mathbf{P^*\tilde{A}}||^2_F.
\end{equation}
$\mathbf{\tilde{A}}$ is a projection-cost preservation sketch of $\mathbf{A}$, so for the projection matrix $\mathbf{P}^*$, the following equation holds
\begin{equation}
\label{eq:lemma2}
||\mathbf{\tilde{A}}-\mathbf{P^*\tilde{A}}||^2_F+c\leq (1+\epsilon) ||\mathbf{A}-\mathbf{P^*A}||^2_F.
\end{equation}
Again, consider projection matrix $\mathbf{\tilde{P}^*}$, we have
\begin{equation}
\label{eq:lemma3}
||\mathbf{A}-\mathbf{\tilde{P}^*A}||^2_F\leq ||\mathbf{\tilde{A}}-\mathbf{\tilde{P}^*\tilde{A}}||^2_F+c.
\end{equation}
Combining equation \ref{eq:lemma1} and \ref{eq:lemma3}, we get 
\begin{equation}
\label{eq:lemma4}
||\mathbf{A}-\mathbf{\tilde{P}^*A}||^2_F \leq ||\mathbf{\tilde{A}}-\mathbf{P^*\tilde{A}}||^2_F+c.
\end{equation}
Finally, combining equation \ref{eq:lemma2} and \ref{eq:lemma4}, it yields
$$
||\mathbf{A}-\mathbf{\tilde{P}^*A}||^2_F \leq (1+\epsilon)||\mathbf{A}-\mathbf{P^*A}||^2_F.
$$
\end{proof}
The above lemma provides us a theoretical guarantee to get meaningful information from a matrix sketch, which is computationally efficient. To capture the graph information related to normalized cut, we consider Equation \ref{eq:eigen_origin}. Solving this generalized eigenvector problem is not convenient, since it's not easy to compute incrementally and more importantly it's not a constrained low rank approximation problem. Therefore, we transform it to a SVD problem by setting $\mathbf{z}=\mathbf{D}^{1/2}\mathbf{x}$ in equation \ref{eq:eigen_origin}, then we have
$$
(\mathbf{I}-\mathbf{D^{-1/2}WD^{-1/2}})\mathbf{z}=\lambda \mathbf{z}.
$$
Removing non-relevant terms, we actually want to find the top eigenvectors of $\mathbf{D^{-1/2}WD^{-1/2}}$ which we denote as $\boldsymbol{\mathcal{L}}$. $\boldsymbol{\mathcal{L}}$ is actually called normalized random walk matrix. If we have a matrix sketch $\mathbf{M}\in \mathbb{R}^{n\times d}$ of $\boldsymbol{\mathcal{L}}$, one can use top singular vectors of $\mathbf{M}$ to approximate top eigenvectors of $\boldsymbol{\mathcal{L}}$ since $\mathcal{L}$ is a symmetric matrix. 

Next, we will decouple the nodes' relations in the graph in terms of $\boldsymbol{\mathcal{L}}$ with preserving projection-cost. One can regard each node as a row vector in $\boldsymbol{\mathcal{L}}$, which is generally sparse. To remove the connection between nodes, we randomly project their vectors onto orthogonal directions $d$ times so that nodes' vectors are mapped into that orthogonal subspace. More precisely, for a node vector $\mathbf{v}$, we choose a map $g(\mathbf{v})=\frac{1}{\sqrt{d}}\mathbf{Rv}$, where $\mathbf{R}\in \mathbb{R}^{d\times n}$ and $\mathbf{RR}^T$ is an orthogonal projection that maps vectors into a uniformly random subspace of dimension $d$. This strategy should work fine if we have this random matrix $\mathbf{R}$. However, ensuring the orthogonality of the projection matrix $\mathbf{RR}^T$ takes unnecessary time. We can achieve the same goal without explicitly orthonormalize the projection matrix. We choose $\mathbf{R}$ to be a Johnson-Lindenstrauss matrix, that is, the entries of $\mathbf{R}$ are independently and uniformly drawn from Gaussian distribution $N(0,1)$. In this way, although the eigenvalues of $\mathbf{RR}^T$ are not confined in $\{0,1\}$, the range of $\mathbf{RR}^T$ is indeed a uniformly random subspace. In a matrix form, it means that
\begin{equation}
\mathbf{M}=g(\boldsymbol{\mathcal{L}})=\frac{1}{\sqrt{d}}\boldsymbol{\mathcal{L}}\mathbf{R^T}.
\end{equation}
Now, we have got a matrix sketch by Johnson-Lindenstrass random projection. The following lemma indicates that the matrix sketch $\mathbf{M}$ generated by this procedure is indeed a projection-cost preserving sketch.
\begin{lemma}
\label{lemma:random_projection}
For matrix $\mathbf{A}\in \mathbb{R}^{n\times n}$, let  $\mathbf{R}\in \mathbb{R}^{d\times n}$ be a Johnson-Lindenstrauss matrix with each entry chosen independently and uniformly from Gaussian distribution $N(0,1)$. For $\epsilon>0$, with probability at least $1-2/n$, $\frac{1}{\sqrt{d}}\mathbf{A}\mathbf{R}^T$ is a rank $k$ projection-cost preserving sketch of $\mathbf{A}$ with approximation ratio $\epsilon$, when $d\geq k/\epsilon^2$.
\end{lemma}

By lemma \ref{lemma:random_projection}, one can achieve an accurate sketch with a satisfied approximation ratio by increasing the sketch size $d$. We will talk about how to choose this parameter later.


\subsection{Feature Extraction by Low Rank Approximation}
After obtaining the projection-cost preserving sketch $\mathbf{M}$ of $\boldsymbol{\mathcal{L}}$, we want to extract meaningful information from the sketch $\mathbf{M}$ and also further reduce its dimension. Singular value decomposition is a good choice, since it is a constrained low rank approximation problem which is suitable and advantageous for further factorizing projection-cost preserving sketch, and it is easy to adapt to unseen nodes. A partial singular value decomposition over matrix $\mathbf{M}$ will give
\begin{equation}
\mathbf{M_k}=\mathbf{\tilde{U}_k}\mathbf{\tilde{\Sigma}_k}\mathbf{\tilde{V}}^T_k.
\end{equation}
Each row of $\mathbf{D}^{-\frac{1}{2}}\mathbf{\tilde{U}_k}$ is the learned graph representation for the corresponding node in the graph. Further, to demonstrate the effectiveness of this framework to learn accurate low-dimensional vectors and to facilitate important tasks (e.g., $k$-means clustering), we present the following theorem.
\begin{theorem}
\label{theorem:difference}
Let $\mathbf{MM}^T = \mathbf{\mathcal{LL}}^T+\mathbf{\Delta}$. And $\mathbf{\mathcal{L}}$ has a singular value decomposition as $ \mathbf{\mathcal{L}} = \mathbf{U\Sigma V^T}$ with distinct singular values. Let $\eta=\frac{||\mathbf{\Delta}||_F}{\rho}$ and $\alpha=\frac{1+\sqrt{1-\frac{1}{n}}}{\rho}$. If $\eta\leq \frac{1}{2\alpha+\sqrt{1+4\alpha^2}}$, then $\mathbf{M}$ has the singular value decomposition $\mathbf{M}=\tilde{\mathbf{U}}\mathbf{\tilde{\Sigma}}\tilde{\mathbf{V}}^T$, such that
$
||\tilde{\mathbf{U}}-\mathbf{U}||_F\leq \gamma
$
where $\gamma = \frac{\sqrt{2}\eta}{\sqrt{1-2\alpha\eta+\sqrt{1-4\alpha\eta-\eta^2}}}$ and $\rho=\underset{1\leq i\neq j\leq n}{\operatorname{min}} |\sigma_i-\sigma_j|$.
\noindent
Furthermore, $k$-means clustering over $\tilde{\mathbf{U}}$ will give a good approximate to that over $\mathbf{U}$ in terms of $k$-means centroids.
\end{theorem}

\begin{proof}
For the first part of the proof, since  $\mathbf{MM}^T$ and $\boldsymbol{\mathcal{LL}}^T$ are symmetric matrices, we can view $\mathbf{\Delta}$ as a symmetric perturbation matrix. Deriving the upper bound of $||\tilde{\mathbf{U}}-\mathbf{U}||_F $ is similar to the derivation of the absolute perturbation bound for eigenvector decomposition in perturbation theory \cite{chen2012perturbation}.

The key of the second part is to regard $k$-means clustering as a constrained low rank approximation problem \cite{boutsidis2009unsupervised}. As defined in Section \ref{sec:clra}, $\mathbf{XX}^T$ is a projection matrix projecting the points vector into its cluster centroid and $\mathbf{X}$ is the cluster indicator matrix. The discrepancy between the corresponding cluster position assigned to each node of two matrix $\tilde{\mathbf{U}}$ and $\mathbf{U}$ is $||\mathbf{XX}^{T}(\tilde{\mathbf{U}}-\mathbf{U})||_F$. Applying the spectral submultiplicativity property, it yields
$$
||\mathbf{XX}^{T}(\tilde{\mathbf{U}}-\mathbf{U})||_F\leq ||\tilde{\mathbf{U}}-\mathbf{U})||_F^2 ||\mathbf{XX}^{T}||_2.
$$
Since $\mathbf{XX}^{T}$ is a symmetric projection matrix, the spectral norm of $\mathbf{XX}^{T}$ is not greater than 1. Then, it becomes $$
||\mathbf{XX}^{T}(\tilde{\mathbf{U}}-\mathbf{U})||_F\leq \gamma.
$$
It means that $\mathbf{\tilde{U}}$ approximates nodes' representations in $\mathbf{U}$ well and the assigned clusters for nodes in $\mathbf{U}$ and $\mathbf{\tilde{U}}$ are well-matched. 
\end{proof}

Generally, this theorem states that the obtained representations of our algorithm will give a strong guarantee if we perform $k$-means clustering on them. In the experiment, we demonstrate that FI-GRL also achieve excellent results on clustering task using agglomerative method (AM) and on structural hole detection task.
\subsection{Inductive Learning and Entire Framework of FI-GRL}
\label{subsec:entire_framework}
So far, we've presented the graph representation learning framework for static graphs. When considering inductive learning on an unseen node $v_i$, we first get $\boldsymbol{\mathcal{L}}_i$ of the corresponding column vector of the normalized random walk matrix of the extended graph after $v_i$ added. The valid dimension of $\boldsymbol{\mathcal{L}}_i$ is at most $n$ since self-join is prohibited. Then, applying the random projection matrix $\mathbf{R}$ on $\boldsymbol{\mathcal{L}}_i$, we get a compressed vector $\mathbf{b}=\frac{1}{\sqrt{d}}\mathbf{R}\boldsymbol{\mathcal{L}}_i$, where $d$ is sketch size as above. Since the partial SVD on matrix $\mathbf{M}$ is $\mathbf{M_k}=\mathbf{\tilde{U}_k}\mathbf{\tilde{\Sigma}_k}\mathbf{\tilde{V}^T_k}$, we can regard row vectors of $\mathbf{\tilde{U}_k}$ as vectors in the span of $\mathbf{\tilde{V}^T_k}$. So we can project $\mathbf{b}$ onto the span of $\mathbf{\tilde{V}^T_k}$. Specifically,
\begin{equation}
\mathbf{\hat{b}}=\mathbf{b}\mathbf{\tilde{V}_k}\mathbf{\tilde{\Sigma}_k^{-1}}.
\end{equation}
Then, degree-normalized $\mathbf{\hat{b}}$, namely $1/\sqrt{\mathbf{D}_{jj}}\mathbf{\hat{b}}$, is the obtained representation of node $v_i$. This method is fast and powerful when the graph is stable and gradually changes. We testify the effectiveness of our method at different proportions of unseen nodes in the experiment. Overall, our FI-GRL framework is summarized in Algorithm \ref{alg:fs-grl}.
\begin{algorithm}[htbp]
\small
\caption{\textbf{FI-GRL}: Fast Inductive Graph Representation Learning }
\label{alg:fs-grl}
\begin{algorithmic}[1] 
\REQUIRE ~Graph $\mathcal{G}=(V,E,\mathbf{W})$ with totally $n$ nodes; Unseen node set $\{v_i\}$;Dimension $k$, approximation ratio $\epsilon$
\ENSURE Low-dimensional vectors $(\mathbf{y}_1,\mathbf{y}_2,\cdots,\mathbf{y}_n,\cdots)$

\STATE Construct matrix $\boldsymbol{\mathcal{L}}=\mathbf{D^{-1/2}WD^{-1/2}}$ for $\mathcal{G}_1$
\STATE Construct a $d\times n$ matrix $\mathbf{R}$, whose entries are independently drawn from $N(0,1)$, where $d$ is $max\{4log(n)/\epsilon ^2,k/\epsilon^2\}$
\STATE Compute each row of the matrix sketch $\mathbf{M}$, $\mathbf{M}_i=\frac{1}{\sqrt{d}}\mathbf{R}\boldsymbol{\mathcal{L}}_i$, where $\boldsymbol{\mathcal{L}}_i$ denotes $i$th row of $\boldsymbol{\mathcal{L}}$
\STATE Compute $k$-singular value decomposition $\mathbf{M_k}=\mathbf{\tilde{U}_k}\mathbf{\tilde{\Sigma}_k}\mathbf{\tilde{V}}^T_k$
\STATE Compute $\mathbf{Y}=\mathbf{D}^{-\frac{1}{2}}\mathbf{\tilde{U}_k}$
\FORALL{unseen nodes $v_j$}
\STATE Compute $\mathbf{b}=\frac{1}{\sqrt{d}}\mathbf{R}\boldsymbol{\mathcal{L}}_j$
\STATE Compute $\mathbf{\hat{b}}=1/\sqrt{\mathbf{D}_{jj}}\mathbf{b}\mathbf{\tilde{V}_k}\mathbf{\tilde{\Sigma}_k^{-1}}$ 

\STATE Append $\mathbf{\hat{b}}$ as a new row of $\mathbf{Y}$
\ENDFOR
\RETURN $\mathbf{Y}$
\end{algorithmic}
\end{algorithm} 
\subsection{Parameter Analysis and Complexity Analysis}
\label{sec:parameter_complex}
The parameter of the sketch size $d$ can actually be determined by the approximation ratio $\epsilon$. The approximation ratio $\epsilon$ can be designated for the requirement of different tasks. $\epsilon=0.1$ will be sufficient for most tasks focusing on mining the macroscopic structure of the graph (e.g., clustering). One can always decrease $\epsilon$ to achieve better accuracy if computational power is allowed. Johnson-Lindenstrass lemma will provide us another perspective to determine the sketch size $d$.
\begin{lemma}
\label{lemma:JL}
Let $\mathbf{x}_1,\cdots,\mathbf{x}_n\in \mathbb{R}^t$ be arbitrary. Pick any $\epsilon \in (0,1)$ and matrix $\mathbf{R}\in \mathbb{R}^{d\times t}$ as a Johnson-Lindenstrauss random projection matrix whose entries are independently and uniformly drawn from Gaussian distribution $N(0,1)$. Then, for $d=O(log(n)/ \epsilon^2)$, define $\mathbf{y}_i=\mathbf{R}\mathbf{x}_j/\sqrt{d}$ for $i=1,\cdots,n$, then for any $j,j'$ the following equations hold with probability $1-2/n:$
\begin{equation}
\begin{aligned}
(1-\epsilon)||\mathbf{x}_j||\leq ||&\mathbf{y}_j||\leq (1+\epsilon)||\mathbf{x}_j||,\\
(1-\epsilon)||\mathbf{x}_j-\mathbf{x}_{j'}||\leq ||\mathbf{y}_j & -\mathbf{y}_{j'}||\leq (1+\epsilon)||\mathbf{x}_j-\mathbf{x}_{j'}||.
\end{aligned}
\end{equation}
\end{lemma}
Taking each row of our matrix $\boldsymbol{\mathcal{L}}$ as $\mathbf{x}_i$, Lemma \ref{lemma:JL} proves that the norm of vectors and the distance between nodes are preserved in the low-dimensional subspace when $d=O(log(n)/ \epsilon^2)$. In fact, $d=\delta log(n)/\epsilon^2$ \cite{dasgupta2000experiments}, where $\delta\leq 4$. When approximation ratio is known, we choose the $max\{4log(n)/\epsilon^2,k/\epsilon^2\}$ as the sketch size $d$.

To analyze the computational complexity of our algorithm, we first note that our algorithm is especially efficient since matrix $\boldsymbol{\mathcal{L}}$ is very sparse when the graph is large, and the matrix-vector product $\mathbf{R}\boldsymbol{\mathcal{L}}_i$ can be evaluated rapidly. The computational cost of our algorithm in static settings are $O(dn+dD+d^2n)$ in total, where $D$ is the total degree of the graph. $O(dn)$ is the cost of generating the Johnson-Lindenstrass random projection matrix. $O(dD)$ is the cost of computing the projection-cost preserving sketch. $O(d^2n)$ is for computing the partial singular value decomposition. For unseen nodes, to learn a node representation of an unseen node $v_i$, we need $O(d\mathbf{D}_{ii}+dk)$ time where $\mathbf{D}_{ii}$ is the degree of $v_i$. $O(d\mathbf{D}_{ii})$ is the time of projecting the node into $d$-dimensional space and $O(dk)$ is the cost of folding-in to the span of right singular vectors.
\subsection{Discussion}
Since matrix $\boldsymbol{\mathcal{L}}$ is symmetric, we are able to use double-sided random projection with projection-cost preserved and use eigenvector decomposition as a generalized constrained low rank approximation. To see this, we compute the projection-cost of a projection $\mathbf{QQ}^T$ on double sides of the symmetric matrix $\boldsymbol{\mathcal{L}}$, then it yields
\begin{equation}
\begin{footnotesize}
\begin{aligned}
&||\boldsymbol{\mathcal{L}}-\mathbf{QQ}^T\boldsymbol{\mathcal{L}}\mathbf{QQ}^T||_F^2=||\boldsymbol{\mathcal{L}}-\mathbf{QQ}^T\boldsymbol{\mathcal{L}}+\mathbf{QQ}^T\boldsymbol{\mathcal{L}}-\mathbf{QQ}^T\boldsymbol{\mathcal{L}}\mathbf{QQ}^T||_F^2\\
&\leq ||\boldsymbol{\mathcal{L}}-\mathbf{QQ}^T\boldsymbol{\mathcal{L}}||_F^2+||\boldsymbol{\mathcal{L}}-\boldsymbol{\mathcal{L}}\mathbf{QQ}^T||_F^2||\mathbf{QQ}^T||_2^2
\leq 2||\boldsymbol{\mathcal{L}}-\mathbf{QQ}^T\boldsymbol{\mathcal{L}}||_F^2.
\end{aligned}
\end{footnotesize}
\end{equation}
\normalsize
The lemmas and theorems in this paper can be devised correspondingly. However, this approach needs an additional computation of the product of the top k eigenvectors of the matrix sketch and the random projection matrix. Therefore, double-sided projection-cost preservation is not necessary and regular projection-cost preservation is sufficient for graph representation learning purpose.


\section{Experimental Results}
\label{sec:results}
To quantitatively testify our FI-GRL framework, we perform various experiments using the learned graph representations. The implementation of our algorithm is publicly available\footnote{https://github.com/Jafree/FastInductiveGraphEmbedding}.
\subsection{Datasets Description and Comparison Methods}
All datasets used in this paper are undirected graphs, which are available in SNAP Datasets platform \cite{snapnets}. These networks vary widely from network type, network scale, edge density, connecting patterns and cluster profiles, which contain three social networks: karate (real), youtube (online), enron-email (communication); three collaboration networks: ca-hepth, dblp, ca-condmat; three entity networks: dolphins (animals), us-football (organizations), polblogs (hyperlinks). To show the characteristics of these datasets, we use a community detection algorithm, RankCom \cite{jiang2015fast}, designed for graphs to reveal the cluster profiles (cluster numbers and max size of clusters). The detailed information is summarized in Table \ref{tab:datasets}.
\begin{table}[htbp] 
\scriptsize
\caption{\bf Summary of datasets and their cluster profiles.}
\centering 
\begin{tabular}{|C{1.55cm}|C{1cm}C{1cm}C{1cm}C{1.8cm}|}
\hline
\multirow{2}{*}{}&\multicolumn{2}{c}{\textbf{Characteristics}}&\textbf{\#Cluster}&\textbf{\#Max members}\\
\cline{2-5}Datasets&\# Node&\# Edge&\emph{RankCom}&\emph{RankCom}\\
\hline
\emph{karate}&34&78&2   & 18  \\

\emph{dolphins}&62&159&3   & 29 \\

\emph{us-football}&115&613&11   & 17  \\

\emph{polblogs}&1,224&19,090&7  & 675  \\

\emph{ca-hepth}&9,877&25,998&995  & 446  \\

\emph{ca-condmat}&23,133&93,497&2,456  & 797  \\

\emph{email-enron}&36,692&183,831&3,888  & 3,914  \\

\emph{youtube}&334,863&925,872&15,863  & 37,255 \\

\emph{dblp}&317,080&1,049,866&25,633  & 1,099 \\

\hline
\end{tabular}
\label{tab:datasets}

\end{table}

Our framework is testified against state-of-the-art algorithms. The first five are graph representation learning methods. The others are structural hole detection methods . We summarize them as follows:
\begin{itemize}
\item \textbf{FI-GRL}: Our inductive representation learning approach.
\item \textbf{GraphSAGE} \cite{hamilton2017inductive}: Sampling and aggregating strategy is applied to integrate neighbors' information.
\item \textbf{Deepwalk} \cite{perozzi2014deepwalk}: Truncated random walk and language modeling techniques are adopted to learn representations.
\item \textbf{node2vec} \cite{grover2016node2vec}: Skip-gram framework is extended to networks.
\item \textbf{LINE} \cite{tang2015line}: The version of combining first-order and second-order proximity is used here. 

\item \textbf{HAM} \cite{he2016joint}: A harmonic modularity function is presented to tackle the structural hole detection problem.

\item \textbf{Constraint} \cite{burt2009structural}: Constructing a constraint to prune nodes without certain connectivity.

\item \textbf{Pagerank} \cite{page1999pagerank}: The assumption that structural holes are nodes with high pagerank score is adopted .

\item \textbf{Betweenness Centrality} (BC) \cite{brandes2001faster}: Nodes with highest BC will be selected as structural holes.

\item \textbf{HIS} \cite{lou2013mining}:Optimizing the provided objective function by  a two-stage information flow model .

\item \textbf{AP\_BICC} \cite{rezvani2015identifying}: This method is designed by exploiting the approximate inverse closeness centralities.
\end{itemize}

\subsection{Clustering}
\begin{table*}[htbp]

\scriptsize
\caption{
\bf{Performance on Clustering evaluated by Modularity and Permanence$_{(rank)}$}}
\centering 
\begin{tabular}{C{1cm}|C{1cm}||C{1cm}C{1cm}C{1cm}C{1cm}C{1cm}||C{1cm}C{1cm}C{1cm}C{1cm}C{1cm}}
\hline
\multirow{2}{*}{}&&\multicolumn{5}{c}{Modularity}&\multicolumn{5}{c}{Permanence}\\

\cline{3-12}Datasets& Clustering Methods&\textbf{FI-GRL}&\textbf{node2vec}&\textbf{GraphSAGE}&\textbf{LINE}&\textbf{Deepwalk}&\textbf{FI-GRL}&\textbf{node2vec}&\textbf{GraphSAGE}&\textbf{LINE}&\textbf{Deepwalk} \\

\hline
\multirow{2}{*}{\emph{karate}}&\emph{k-means}& \textbf{0.410}(1)& 0.335(5)&0.381(4)&  0.403(2)&0.396(3)&  \textbf{0.474}(1)& 0.335(3)& 0.322(4)& 0.182(5)&0.350(2)  \\
&\emph{AM}& 0.410(2)&  0.335(4)&0.401(3)&  0.239(5)& \textbf{0.430}(1)& \textbf{0.474}(1)&  0.205(5)&0.339(2)&  0.232(4)& 0.311(3)\\
\cline{1-2}

\multirow{2}{*}{\emph{dolphins}}&\emph{k-means}&  \textbf{0.489}(1)& 0.460(2)&  0.370(4)&0.187(5)&0.401(3) &  \textbf{0.235}(1)& 0.196(2)&0.158(4)&  -0.166(5)&0.187(3)\\
&\emph{AM}& \textbf{ 0.462}(1)& 0.458(2)&  0.355(4)&0.271(5)& 0.393(3)&\textbf{0.215}(1)& 0.132(3)&0.121(4)&  -0.189(5)& 0.189(2)\\
\cline{1-2}

\multirow{2}{*}{\emph{us-football}}&\emph{k-means}& \textbf{0.607}(1)& 0.605(2)&0.485(4)& 0.562(3) & 0.464(5)& \textbf{0.323}(1) & 0.304(3)&0.124(4)& 0.311(2) & 0.039(5)\\
&\emph{AM}&  \textbf{0.611}(1)& 0.589(2)&0.470(4)& 0.492(3) & 0.464(5)& \textbf{0.315}(1)& 0.279(3)&0.116(4)& 0.307(2) & 0.039(5)\\
\cline{1-2}

\multirow{2}{*}{\emph{ca-hepth}}&\emph{k-means}&  \textbf{0.611}(1)& 0.597(2)&0.399(4)& 0.01(5)& 0.424(3)& \textbf{0.393}(1)& 0.379(2)&0.287(3)& -0.948(5) & 0.261(4)\\
&\emph{AM}&  \textbf{0.623}(1)& 0.606(2)&0.423(4)& 0.05(5)& 0.453(3)& \textbf{0.427}(1)& 0.406(2)&0.327(4)& -0.949(5) & 0.338(3)\\
\cline{1-2}

\multirow{2}{*}{\emph{condmat}}&\emph{k-means}& \textbf{0.527}(1)& 0.515(2) &0.409(3)&  0(5)& 0.357(4)& \textbf{0.371}(1)& 0.330(2)&0.206(3)&  -0.984(5)& 0.197(4)\\
&\emph{AM}& \textbf{0.544}(1)& 0.520(2)&0.427(3)& 0(5) & 0.370(4)& \textbf{0.392}(1)& 0.388(2)&0.213(4)& -0.994(5) & 0.249(3)\\
\cline{1-2}

\multirow{2}{*}{\emph{enron-email}}&\emph{k-means}&  \textbf{0.322}(1)& 0.213(3)&0.231(2)&  0(5)& 0.178(4)& \textbf{0.175}(1)& 0.080(2)&0.067(3)&  -0.985(5)& 0.049(4)\\
&\emph{AM}&  \textbf{0.327}(1)& 0.218(2)&0.211(3)& 0(5) & 0.207(4)&  \textbf{0.187}(1)& 0.180(2)&0.058(4)& -0.996(5) & 0.108(3)\\
\cline{1-2}

\multirow{2}{*}{\emph{polblogs}}&\emph{k-means}& \textbf{0.427}(1)& 0.357(2)&0.278(3)&  0.200(4)& 0.084(5)& \textbf{0.130}(1)& -0.066(2)&-0.106(3)&  -0.569(5)& -0.187(4)\\
&\emph{AM}& \textbf{0.425}(1)& 0.376(2)&0.291(3)&  0.266(4)& 0.065(5)& \textbf{0.131}(1)& -0.096(2)&-0.123(3)&  -0.509(5)& -0.176(4)\\
\cline{1-2}

\hline
\end{tabular}
\label{tab:community_result}
\end{table*}

We first test our algorithm on the clustering task. We set the approximation ratio $\epsilon=0.1$, and set the dimension of representations at most $200$ for all graphs. Firstly, we perform our FI-GRL algorithm on graphs to learn low-dimensional representations. Then, two different type of clustering algorithms, i.e., $k$-means clustering and agglomerative clustering method (AM) are applied. We use two different metrics to evaluate the clustering results, i.e., modularity  \cite{newman2006modularity} and permanence \cite{chakraborty2014permanence}.
\begin{itemize}
\item \textbf{Modularity} \cite{newman2006modularity}: This is the most widely used metric for evaluating clustering results on graphs. Modularity measures the benefits of nodes joining a cluster under the Null model. Specifically, modularity is defined as:
$
Q=\frac{1}{2m}\sum_{vw}[\mathbf{W}_{vw}-\frac{\mathbf{D}_{vv}\cdot \mathbf{D}_{ww}}{2m}]\delta (c_v,c_w),
$
where $\delta$ is the indicator function. $c_v$ indicates the cluster node $v$ belongs to. Generally, a modularity score greater than $0.3$ means a good clustering result. To punish clearly wrong cluster membership assignment, we add a penalty which is proportional to the inverse of the node's degree.

\item \textbf{Permanence} \cite{chakraborty2014permanence}: Permanence is a node-based metric, which explicitly evaluate the cluster membership affiliation of each node. It is more strict, since it considers the cluster configuration nodes connecting to. For a node $v$ in cluster $c$ , the permanence is defined as follows:
$
Perm_c(v) = [\frac{I_c(v)}{E^c_{max}(v)}\times \frac{1}{\mathbf{D}_{vv}}]-[1-C^c_{in}(v)],
$
where $I_c(v)$ is the internal degree, $E^c_{max}(v)$ is the maximum degree that node $v$ links to another cluster, and $C^c_{in}(v)$ is the internal clustering coefficient. The total permanence score of the graph is the sum of the permanence score of every node. Empirically, positive permanence score indicates a good clustering result.
\end{itemize}

The results are listed in Table \ref{tab:community_result}. Our algorithm FI-GRL outperforms other graph representation learning algorithms, i.e., GraphSAGE, node2vec, LINE, Deepwalk over almost all datasets using $k$-means and AM in terms of both modularity and permanence (except on karate network, deepwalk outperforms FI-GRL in terms of modularity under AM). Specifically,node2vec achieves the second best results. It performs badly on small networks, i.e., karate in terms of modularity and karate, dolphins, us-football in terms of permanence. LINE fails on capturing the macroscopic structure of the graph since it only preserves local information. Some zero results of LINE on modularity means that there are a bunch of nodes that LINE assigns to a clearly wrong cluster. GraphSAGE and Deepwalk give mediocre results.  In terms of permanence, all other methods cannot preserve the cluster information on polblog, which is a tough case.
The overall performance is reported in Figure \ref{fig:clustering_performance_bar}. More precisely, combining the results using $k$-means and AM, in terms of modularity our algorithm FI-GRL improve node2vec by $9\%$, GraphSAGE by $36\%$, LINE by $153\%$ and Deepwalk by $45\%$. FI-GRL gets an improvement of $39\%$ over node2vec, $95\%$ over GraphSAGE and $115\%$ over Deepwalk in term of permanence. 
\vspace{-10pt}
\begin{figure}[htbp]
\begin{center}
\includegraphics[width=\linewidth]{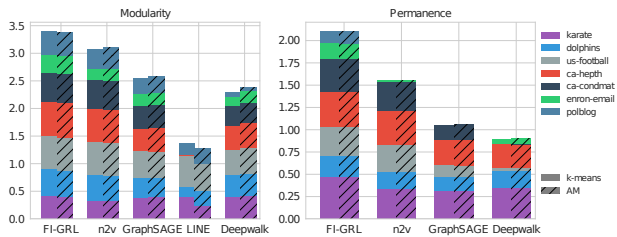}
\end{center}
\caption{
{ The overall performance on Clustering in terms of Modularity and Permanence.}
}
\label{fig:clustering_performance_bar}
\end{figure}

\subsection{Structural Hole Detection}
We consider another task, which is focusing on the microscopic level, called structural hole detection. Structural holes are the important nodes that locate at key topological positions. Once they are removed, the network will fall apart. Finding structural holes in graphs is a critical task for graph theory and information diffusion. To achieve this task, we first transform graph into a low-dimensional subspace using our algorithm, and then find structural holes in that space. We devise a metric for ranking nodes in the low-dimensional subspace:
\begin{itemize}
\item \textbf{Relative Deviation Score} (RDS): Let $\mathbf{y}_v \in \mathbb{R}^k$ be the low-dimensional representation for node $v$. $k$-means will give a clustering result with cluster set  $\mathcal{C}$. RDS estimates the deviation of a node from its own cluster attracted by other clusters in terms of relative radius. More precisely,
$
RDS(v)=\max_{C\in \mathcal{C}} \frac{\parallel \mathbf{y}_v-\mathbf{u}_{C_v}\parallel_2/R_{C_v}}{\parallel\mathbf{y}_v-\mathbf{u}_{C}\parallel_2/R_C}
$, where $C_v$ is the cluster that $v$ belongs to. $R_C=\sum_{i\in C}\parallel \mathbf{y}_i - \mathbf{u}_C\parallel_2$ represents the radius of cluster $C$. And $\mathbf{u}_C=\frac{1}{|C|}\sum_{i\in C} \mathbf{y}_i$ is the center of cluster $C$.
\end{itemize}

Nodes with highest RDS scores are regarded as the structural holes since they strongly connect at least two clusters. We use an evaluation metric called Structural Hole Influence Index (SHII) \cite{he2016joint} to evaluate the selected structural holes. SHII is computed via a process of information diffusion. For each selected structural hole, we run the information diffusion under linear threshold model (LT) and independent cascade model (IC) 10000 times to get average SHII score. The SHII score is defined as follows:

\begin{itemize}
\item \textbf{Structural Hole Influence Index} \cite{he2016joint}: Note that generally a node cannot activate the influence maximization process by itself. For a selected structural hole $v$, we want to do the following procedure several times: combining $v$ with some randomly chosen node set $S_v$ in cluster $C_v$ as a seed set to engage a influence maximization process in the network. SHII evaluate the ratio of activated nodes that are in other clusters $
SHII(v,S_v)=\frac{\sum_{C_i\in \mathcal{C} \backslash C_v} \sum_{u\in C_i}I_u}{\sum_{u\in C_v}I_u}$, where $\mathcal{C}$ is the set of communities and $I_u$ is the indicator function indicating whether node $u$ is influenced. And in our experiment, we set the size of the sampled activation set $|S_v|$ as $0.1|C_v|$. 
\end{itemize}

The results are shown in Table \ref{tab:sh}. According to characteristics of different networks, we tune all algorithms to select a certain number of structural holes. Too many of them will result in the activation of the entire network. For karate network, three structural holes are selected. The topological structure of karate shown in \cite{zachary1977information} demonstrates that the structural holes our algorithm selected are in critical positions that are bridging two clusters. More precisely, our results are superior to other structural hole selection methods, including the state-of-the-art algorithm, HAM \cite{he2016joint}. It demonstrates the efficacy of our algorithm in preserving microscopic structure.

\begin{table*}[htbp]
\scriptsize

\caption{
\bf{Performance on Structural Hole Detection under LT and IC Models}}
\centering 
\begin{tabular}{cccccccccc}
\hline
\multirow{3}{*}{}& &\multicolumn{7}{c}{Comparative Methods}\\
\cline{4-10}
Datasets&\#SH&Influence Model&FI-GRL&HAM&Constraint&PageRank&BC&HIS&AP\_BICC \\
\hline
\multirow{3}{*}{karate}&\multirow{3}{*}{3}&LT&\bf 0.595&0.343&0.295&0.159&0.159&0.132&0.295  \\
& &IC&\bf 0.003&0.002&0.002&0.001&0.001&0.001&0.002  \\
\cline{3-10}
& &Structural Holes&[3 14 20]&[3 20 9]&[1 34 3]&[34 1 33]&[1 34 33]&[32 9 14]&[1 3 34]\\
\hline
\multirow{2}{*}{youtube}&\multirow{2}{*}{78}&LT&\bf 4.129&3.951&2.447&1.236&1.226&3.198&1.630  \\
& &IC&\bf 3.024&2.452&1.254&0.662&0.791&2.148&0.799  \\
\cline{3-10}
\multirow{2}{*}{dblp}&\multirow{2}{*}{42}&LT&\bf 6.873&5.384&0.404&0.357&0.958&0.718&0.550  \\
& &IC&\bf 5.251&3.578&0.229&0.190&0.821&0.304&0.495  \\
\hline
\end{tabular}
\label{tab:sh}
\end{table*}

\subsection{Performance on Unseen Nodes}
\label{sec:perfor_streaming}
To evaluate the performance of our algorithm under the inductive learning scenario, we artificially simulate the process of generating unseen nodes. Specifically, for a static graph, we randomly extract $\alpha$ proportion of nodes as an original graph for graph representation learning. The other nodes are treated as unseen nodes. In this experiment, we set the approximation ratio as $\epsilon=0.1$, which is accurate enough for most applications. FI-GRL uses a folding-in technique to learn the meaningful representations for unseen nodes. If the learned representations are accurate, $k$-means clustering results over the representations of the entire graph will be satisfactory. 

The clustering performance in terms of modularity with the variation of the proportion of unseen nodes (i.e., $1-\alpha$) is illustrated in Figure \ref{fig:streaming_performance}. As we can see, when the proportion of unseen nodes is not greater than $40\%$, the clustering performance is stable at a good quality. After increasing the proportion to $50\%$, too many unseen nodes added are dramatically changing the main skeleton of the network. Since we add a penalty to clearly wrong cluster assignment, the clustering performance degenerates sharply. For a small network like dolphins, the results fluctuate to a certain extent, e.g., at $10\%$ and $30\%$. In fact, polblog and football give the most stable performance, as they retain almost the same results from $5\%$ to $40\%$. We conjecture that the representations of the unseen nodes are more accurate if the nodes are well-connected (polblog has a relatively high edge density) or the network are well-structured and invulnerable (football has 11 clusters with nearly equal size). Overall, our FI-GRL  is flexible enough to give a satisfactory representation learning result for inductive learning even when the proportion of the unseen nodes is large (up to $40\%$).

\vspace{-10pt}
\begin{figure}[htbp]
\begin{center}
\includegraphics[width=0.8\linewidth]{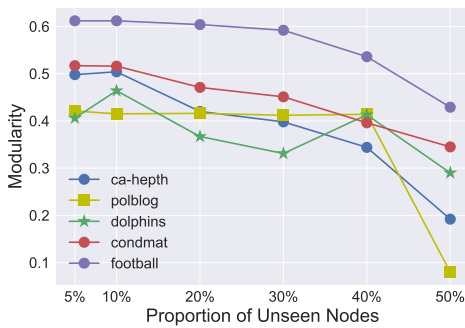}
\end{center}
\caption{
{Performance on unseen nodes evaluated on clustering}
}
\label{fig:streaming_performance}
\end{figure}

\subsection{Parameter Analysis: Approximation Ratio and Sketch Size}
\label{sec:para_empirical_sketch}
To quantitatively measure the ability of our framework to capture crucial information of graphs and preserve the projection-cost, we evaluate the performance of our algorithm by varying the approximation ratio and the sketch size. To get a visual sense, we plot three-dimensional and two-dimensional representations of the karate network \cite{zachary1977information}, which has two ground-truth clusters and several structural holes, in Figure \ref{fig:karate_plot}, when we set the approximation ratio $\epsilon=0.3$ and $\epsilon=0.1$, separately. As we can see, at $\epsilon=0.3$, nodes between two clusters are mixed up with each other. So at a low resolution, the cluster information is not well-maintained in the learned subspace. While at $\epsilon=0.1$, nodes are located in clusters exactly the same as the ground truth. Moreover, structural holes bridging between two clusters can be easily identified from the two-dimensional view where nodes in different clusters form nearly orthogonal subspaces, and they are linearly separable.
\vspace{-10pt}
\begin{figure}[htbp]
\begin{center}
\includegraphics[width=0.8\linewidth]{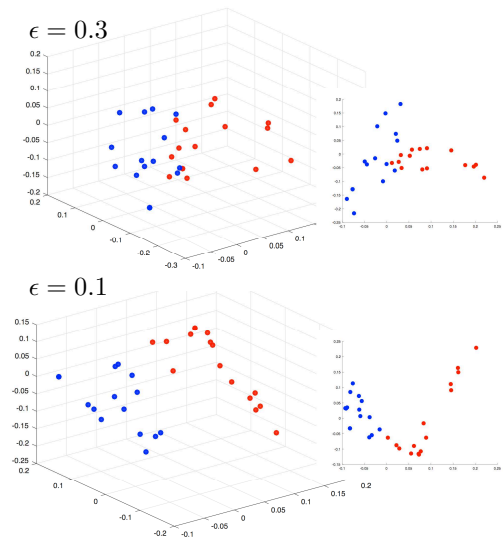}
\end{center}
\caption{
{Visualization at different approximation ratios}
}
\label{fig:karate_plot}
\end{figure}

To demonstrate the ability of our algorithm in preserving projection-cost, we compute the relative projection-cost with the variation of the sketch size. More precisely, we calculate $$(||\boldsymbol{\mathcal{L}}-\mathbf{\tilde{U}}_k\mathbf{\tilde{U}}_k^T\boldsymbol{\mathcal{L}}||_F^2-||\boldsymbol{\mathcal{L}}_{\backslash k}||_F^2)/||\boldsymbol{\mathcal{L}}_{\backslash k}||_F^2,$$ where $\boldsymbol{\mathcal{L}}_{\backslash k}$ is the residual of optimal rank $k$ approximation $\boldsymbol{\mathcal{L}}_{k}$ on $\boldsymbol{\mathcal{L}}$. The dimension $k$ is set to $min\{0.1n,200\}$, which is sufficient for applications we concern. We perform our algorithm $10$ times at each sketch size and the result is shown in Figure \ref{fig:projection_cost}. Relative projection-cost has decreased rapidly at very small sketch size. At sketch size of $400$, FI-GRL already can achieve excellent results. Towards sketch size of $1000$, the result is nearly optimal for graph representation learning purpose. For large networks, the approximation is even more accurate. Since the network is usually very sparse, nodes are laying in a small subspace compared to the size of the network. Although our algorithm is a randomized algorithm, the variance of at each sketch size is rather small. It implies that we can treat FI-GRL as a deterministic algorithm since the chance of the failure of our algorithm in preserving projection-cost is pretty rare especially when the sketch size is large.

\begin{figure}[htbp]
\begin{center}
\includegraphics[width=\linewidth]{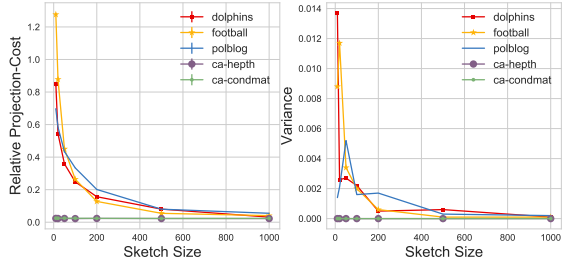}
\end{center}
\caption{
{ Relative Projection-Cost with the variation of Sketch Size}
}
\label{fig:projection_cost}
\end{figure}

\subsection{Running Time}
\label{sec:runing_time}
Finally, the computational time of our algorithm in the static scenario against other competing graph representation algorithms is listed in Table \ref{tab:time}. All other algorithms learn the representations of $100$ dimension. FI-GRL is outstanding in terms of computational cost. At sketch size of $1000$ where FI-GRL can give nearly optimal results for tested datasets, it takes only $two$ minutes to learn the graph representations on dblp. While GraphSAGE, node2vec and deepwalk take more than 10 hours in order to achieve the same task.

\begin{table*}[!tb]
\footnotesize
\caption{
\bf{Running time}}
\centering 
\begin{tabular}{c||ccccccccc}
\hline

 Methods (Sketch Size)&\emph{karate}&\emph{dolphins}&\emph{us-football}&\emph{polblog}&\emph{ca-hepth}&\emph{ca-condmat}&\emph{email-enron}&\emph{youtube}&\emph{dblp}\\
\hline

FI-GRL(100)&0.005s&0.007s&0.042s&0.032s&0.181s&0.424s&0.672s&7.14s&8.33s     \\

FI-GRL(200)&0.006s&0.009s&0.051s&0.064s&  0.388s &0.913s&1.508s&15.85s&18.09s\\

FI-GRL(500)&0.019s&0.015s&0.037s&0.081s&   0.976s&2.804s&6.776s&45.73s&57.30s\\

FI-GRL(1000)& 0.040s&0.035s&0.047s&0.159s& 2.576s&7.262s&11.48s&1m41s&2m14s\\
\hline
\hline
node2vec&0.807s&3.110s&1.442s&33.34s&74.83s&2m57s&48m17s&$>$10h&$>$10h\\
Deepwalk&4.123s&10.876s&10.92s&2m10s&15m59s&43m9s&1h18m&$>$10h&$>$10h\\
GraphSAGE&18.348s&43.791s&37.252s&6m3s&49m20s&3h51m&4h39m&$>$10h&$>$10h\\

\hline
\end{tabular}
\label{tab:time}
\end{table*}

\section{Related Work}
\label{sec:work}
\subsection{Graph Representation Learning}
Graph representation learning has been an important problem to facilitate the implementation of classic machine learning and data mining algorithms on graphs. Some methods try to explicitly preserve proximity between nodes, such as 
\cite{tang2015line} introduces an edge-sampling method, \cite{wang2016structural} develops a semi-supervised deep model, \cite{jiang2018spectral} enhances communities and structural holes by non-backtracking random walk. Some methods exploit matrix factorization technique, e.g.,  \cite{ou2016asymmetric} factorizes asymmetric transitivity related matrices on directed graphs, \cite{yang2017fast} proposes an update algorithm on matrix forms.
Some algorithms are formulating the problem into a traditional machine learning approach, such as, Deepwalk \cite{perozzi2014deepwalk} learns latent representations by treating truncated walks as sentences, \cite{tu2016max} optimizes a max-margin classifier.
Several methods focus on heterogeneous scenario, such as \cite{ma2017multi} models the multi-view graph data as tensors, \cite{cavallari2017learning} learns the representations of clusters, \cite{tu2017cane} learns context-aware representations, \cite{chang2015heterogeneous} creates a multi-resolution deep architecture, \cite{yang2015network} formulates a Deepwalk-based matrix factorization with incorporating text features, \cite{dong2017metapath2vec} introduces metapath-based random walks for representation learning.

One line of work that are similar to our approach is graph representation learning on dynamic networks. 
\cite{zhou2018dynamic} investigates the role of closed triads at different time steps.
\cite{li2017attributed} integrates node attributes by utilizing spectral decomposition and matrix perturbation theory in a dynamic setting.
\cite{ma2018depthlgp} aims at preserving high-order proximity by using nonparametric probabilistic modeling and deep learning, which can be generalize to unseen nodes.
\cite{hamilton2017inductive} presents several types of aggregators for aggregating features from nodes' local neighborhoods.
In contrast to these works, by introducing randomization and approximation strategies, our approach focuses on building a graph representation learning framework that is fast, theoretically guaranteed and can generalize to unseen nodes.

\subsection{Randomized Dimension Reduction}
Randomized algorithms are often adopted in dimension reduction due to its speed and the solid theory supporting it. \cite{halko2011finding} surveys randomized algorithms for low rank approximation and presents several algorithms for address different situations. The algorithm proposed is more robust than Krylov subspace methods for sparse input matrix. \cite{liberty2013simple} adapt a well known streaming algorithm for approximating item frequencies to find the matrix sketch. Combined with SVD and a special update strategy, the proposed algorithm becomes deterministic and computationally competitive. \cite{cohen2015dimensionality} devises a theoretical framework by deriving a series of bounds in terms of required dimensions for applying random row projection, column selection, and approximate SVD, which can used to better solve $k$-means clustering and low rank approximation problem. \cite{fern2003random} uses random projection in a cluster ensemble approach to achieve better and more robust clustering performance.  \cite{bingham2001random} uses random projection technique to deal with text and image data. It empirically demonstrates that random projection yields comparable results compared to conventional deterministic methods (e.g., PCA), but it is computationally significantly less expensive than PCA. \cite{boutsidis2015randomized} presents the first provably accurate feature selection method for $k$-means clustering. Two feature extraction methods using random projection and fast approximate SVD are proposed, which improves upon the existing results in terms of time complexity. Our approach uses state-of-the-art random strategies in graph representation learning and several bounds and theorems are proved to guarantee the performance of the learned representations.

\section{Conclusion}
\label{sec:conclusion}
In this paper, we propose a fast inductive graph representation learning framework, namely FI-GRL, to transform the topological structure of graphs into a low-dimensional space. It explicitly decouples relational information in graphs into a randomized subspace spanned by a random projection matrix. The sketch obtained are much smaller and yet inherits the property associated with the normalized cut by preserving projection-cost. By exploiting the constrained low rank approximation, the dimension of the sketch is further reduced and the compact hidden pattern is finally extracted. The connection between randomized algorithm and graph representation learning is built by thoroughly theoretical analysis. FI-GRL is flexible enough to deal with massive scale graphs and graph with unseen nodes. Overall, our algorithm is fast, easy to implement and theoretically guaranteed. The empirical study demonstrates the superiority of our algorithm on both efficacy and efficiency.

\section*{Acknowledgement}
This work is supported in part by National Key R\& D Program of China through grants 2016YFB0800700, NSF through grants IIS-1526499, IIS-1763325, CNS-1626432, and NSFC 61672313, 61672051, 61872101.


\bibliographystyle{IEEEtran}
\bibliography{figrl}

\end{document}